\documentclass[a4paper,11pt,reqno]{article}

\usepackage{hyperref}
\hypersetup{
  colorlinks   = true, 
  urlcolor     = blue, 
  linkcolor    = Purple, 
  citecolor   = red 
}

\usepackage[margin=3cm]{geometry}

\usepackage{authblk}
\usepackage{amsmath, amssymb, amsbsy, amsthm, amsfonts}
\usepackage[dvipsnames]{xcolor}
\usepackage{booktabs} 
\usepackage{caption} 
\usepackage{subcaption} 
\usepackage{graphicx}
\usepackage{pgfplots}
\usepackage[all]{nowidow}
\usepackage[utf8]{inputenc}
\usepackage[T1]{fontenc}
\usepackage{tikz}
\usetikzlibrary{er,positioning,bayesnet}
\usepackage{multicol}
\usepackage{algpseudocode,algorithm,algorithmicx}
\usepackage{hyperref}
\usepackage[inline]{enumitem} 
\usepackage{upgreek}
\newcommand{\card}[1]{\left\vert{#1}\right\vert}

\newcommand{\comment}[1]{}

\definecolor{blue}{HTML}{1F77B4}
\definecolor{orange}{HTML}{FF7F0E}
\definecolor{green}{HTML}{2CA02C}

\pgfplotsset{compat=1.14}

\DeclareMathOperator{\rank}{rk}

\DeclareMathOperator{\PG}{PG}
\DeclareMathOperator{\weight}{wt}
\DeclareMathOperator{\dd}{d}
\DeclareMathOperator{\HH}{H}

\newcommand{\dH}{\dd_{\HH}}

\newcommand{\F}{\mathbb F}
\newcommand{\Fq}{\mathbb{F}_q}
\newcommand{\N}{\mathbb N}
\newcommand{\Z}{\mathbb Z}
\newcommand{\code}{C_2(\Uppi\sqcup\Upgamma)^\perp}

\newtheorem{theorem}{Theorem}[section]

\newtheorem{proposition}[theorem]{Proposition}

\newtheorem{lemma}[theorem]{Lemma}

\theoremstyle{definition}
\newtheorem{definition}[theorem]{Definition}
\newtheorem{example}[theorem]{Example}
\newtheorem{remark}[theorem]{Remark}

\begin{document}
\title{Moderate Density Parity-Check Codes from Projective Bundles}

\author[1,2]{Jessica Bariffi}
\affil[1]{Institute of Mathematics, University of Zurich, Switzerland} 
\affil[2]{Institute of Communication and Navigation, German Aerospace Center, Germany}

\author[3]{Sam Mattheus}
\affil[3]{Vrije Universiteit Brussel, Brussels, Belgium}

\author[4]{Alessandro Neri}
\affil[4]{Institute for Communications Engineering, TU Munich, Germany}

\author[1]{Joachim Rosenthal}




 

%
%

%
%

%
\maketitle              
\begin{abstract}
A new construction for moderate density parity-check (MDPC) codes using finite geometry is proposed. We design a parity-check matrix for this family of binary codes as the concatenation of two matrices: the incidence matrix between points and lines of the Desarguesian projective plane and the incidence matrix between points and ovals of a projective bundle. A projective bundle is a special collection of ovals which pairwise meet in a unique point.  We determine  minimum distance and dimension of these codes, showing that they have a  natural quasi-cyclic structure. In addition, we analyze the error-correction performance within one round of a modification of Gallager's bit-flipping decoding algorithm. In this setting, our codes have the best possible error-correction performance for this range of parameters. 

\end{abstract}

\section{Introduction}

The close interplay between coding theory and finite geometry has emerged multiple times in the last 60 years, starting from the works of Prange \cite{Prange} and Rudolph \cite{Rudolph}, where they proposed to construct linear codes starting from projective planes. Their idea was to use the incidence matrix of the plane as a generator matrix or as a parity-check matrix of a linear code, showing that the underlying geometry can be translated in metric properties of the corresponding codes. Generalizations of these constructions have been studied since the 70's and are still subject of active research (see \cite{Assmus}). The relations between these two research areas had also a strong impact in the opposite direction. The most striking example is certainly the non-existence proof of a finite projective plane of order $10$ shown in \cite{lam1989non}. This groundbreaking result came -- with the help of a computer -- after that a series of papers analyzed  the binary linear code coming from a putative projective plane of order $10$. 

A very important class of codes which was sensibly influenced by  geometric constructions is given by low-density parity-check (LDPC) codes, which were introduced by Gallager in
his 1962 seminal paper~\cite{Gal}. LDPC codes, as originally proposed,
are binary linear codes with a very sparse parity-check
matrix. This sparsity property is the bedrock of efficient decoding algorithms.   Already Gallager provided two of such algorithms
 whose decoding complexity is linear in the block length. 
However, 
LDPC codes came to fame much later, when in 2001 Richardson, Shokrollahi and
Urbanke~\cite{ri01a2} were able to show that LDPC codes are capable to
approach Shannon capacity in a practical manner. Above authors derived
this result using random constructions of very large and sparse parity-check matrices. Because of these random constructions the performance
of the codes was only guaranteed with high probability and there was
also the practical disadvantage that the storage of a particular
parity-check matrix required a lot of storage space.

There are several design parameters one wants to optimize when
constructing LDPC codes. On the side of guaranteeing that the distance
is reasonably large, it was realized early that it is desirable that
the girth of the associated Tanner graph is large as well. This last
property helps to avoid decoding failures in many decoding algorithms.
Thus, in order to guarantee that an LDPC code had desirable design
parameters, such as a large distance or a large girth of the associated
Tanner graph, some explicit constructions were needed. Already in 1982
Margulis~\cite{Margulis} used group theoretic methods to construct a
bipartite Cayley graph whose girth was large. This line of research
was extended by Rosenthal and Vontobel~\cite{ro00p} using some
explicit constructions of Ramanujan graphs, which have exceptional
large girth.

Maybe the first time objects from finite geometry were used to
construct explicitly some good LDPC codes was in the work of Kou, Lin
and Fossorier~\cite{ko01}. These authors gave four different
constructions using affine and projective geometries over finite
fields which did guarantee that the resulting code had a good distance
and the associated Tanner graph had a girth of at least 6.
Using points and lines in $\mathbb{F}_q^m$ Kim, Peled, Perepelitsa,
Pless and Friedland \cite{Kim04} came up with incidence matrices representing
excellent LDPC codes. In the last 15 years there has been active
research to come up with further explicit constructions of LDPC codes with desirable parameters based on combinatorial structures \cite{johnson2004low,ko01,liu2005ldpc,vontobel2001construction,Vandendriessche2010}.

Moderate-density parity-check (MDPC) codes were first introduced by Ouzan
and Be'ery~\cite{ou09u}. Misoczki, Tillich, Sendrier and
Barreto~\cite{mi13} showed that MDPC codes could still be decoded with
low complexity as long as the row-weight of each row vector of the
parity-check matrix was not much more than the square root of the
length of the code. These authors also showed that MDPC codes are
highly interesting for the use in the area of code based cryptography.
Similar as for LDPC codes, it is an important task to come up with explicit
constructions of MDPC codes where e.g. a good minimum distance can be
guaranteed. Already Ouzan and Be'ery \cite{ou09u} provided a
construction using cyclotomic cosets. Further constructions using
quasi-cyclic codes can be found in~\cite{Janoska,mi13}.

\medskip
This paper adds another  dowel to the theory of error-correcting codes arising from geometric objects. We propose a new construction of linear codes using projective bundles in a Desarguesian projective plane, resulting in  a family of MDPC codes. Concretely, a projective bundle in a projective plane of order $q$ is a collection of $q^2+q+1$ ovals which mutually intersect in a unique point. We consider the incidence structure consisting of the lines of a projective plane together with the ovals of a projective bundle. The incidence matrix of this  structure will serve as a parity-check matrix of the proposed binary codes. We completely determine their dimension and minimum distance for both $q$ even and odd. In addition, we observe that we can design these codes to possess a quasi-cyclic structure of index $2$. As a consequence, their encoding can be
achieved in linear time and implemented with  linear feedback shift
registers. Moreover, also the storage space required is only half their length. 

The main motivation arises from \cite{Tillich}, where the error-correction capability of  the bit-flipping decoding algorithm on the parity-check matrix of an MDPC code was analyzed. There, it was derived that its performance is inversely proportional to the maximum column intersection of the parity-check matrix, which is the maximum number of positions of ones that two distinct columns share. We show indeed that the maximum column intersection of the derived parity-check matrices is the smallest possible for the chosen parameters, implying in turn the best possible performance of the bit-flipping algorithm.

\medskip

The paper is organized as follows: Section \ref{sec:coding} consists of the coding theory background needed in the paper. In particular, we introduce the family of MDPC codes and we recall the result on the performance of the bit-flipping algorithm presented in \cite{Tillich}, which was decisive for the idea of this construction. In Section \ref{sec:planes} we give a brief overview on projective planes, studying the basic properties of codes arising from them.
 Section \ref{section_construction} is dedicated to  the  new proposed MDPC design using projective bundles. Here, we study some of the code properties and we determine its dimension and  minimum distance. The paper is based on the master's thesis of the first author \cite{Bariffi} and in this section we extend the results which were originally stated there. Finally, the goal of Section \ref{section:generalization} is to generalize the results stated in Section \ref{section_construction} in order to have more flexibility in the choice of the parameters. This is done by using several projective bundles instead of only one.

\section{Coding Theory and Moderate Density Parity-Check Codes}\label{sec:coding}

Let us start by briefly recalling  some basics of coding theory. Throughout the paper $q$ will always be a prime power, and we will denote  the finite field with $q$ elements by $\mathbb{F}_q$. The set of vectors of length $n$ over  $\mathbb{F}_q$ will be denoted by $\mathbb{F}_q^n$. 

We consider the \emph{Hamming weight} on $\Fq^n$  defined as 
$$\weight(v):=\card{\lbrace i \in \{1,\ldots,n\} \mid v_i\neq 0\rbrace}.$$
It is well-known that it induces a metric, namely the \emph{Hamming distance} which is given by
$$ \begin{array}{rccl}
\dH: & \Fq^n \times \Fq^n & \longrightarrow & \N \\
& (u,v) & \longmapsto & \weight(u-v).
\end{array}$$

\begin{definition}
    A \emph{$q$-ary linear code} $C$ of length $n$ and dimension $\dim (C) = k$  is a $k$-dimensional linear subspace of $\mathbb{F}_q^n$ endowed with the Hamming metric. The \emph{minimum distance} of $C$ is the minimum among all the possible weights of the non-zero codewords and it is denoted by $\dd(C)$, i.e.
     \begin{align*}
         \dd(C) := \min \lbrace \weight (c) \, | \, c \in C, \, c \not = 0 \rbrace .
     \end{align*}
\end{definition}
A $q$-ary linear code of length $n$ and dimension $k$ will be denoted for brevity by $[n,k]_q$ code,  or by  $[n,k,d]_q$ code if the minimum distance $d$ is known.

Any $[n,k]_q$ code $C$ has a dual code which is defined as 
    \begin{align*}
        C^\perp = \lbrace x \in \mathbb{F}_q^n \, | \, x\cdot c^\top = 0, \, \forall c \in C \rbrace .
    \end{align*}

 A \emph{generator matrix} of an $[n,k]_q$ code $C$ is a matrix  $G\in\Fq^{k\times n}$ whose rows form a basis of $C$. A generator matrix $H \in \mathbb{F}_q^{(n-k)\times n}$ for the dual code $C^\perp$ is called a \textit{parity-check matrix} of $C$. Note that $C$ can also be represented by a parity-check matrix $H$, since it corresponds to its right kernel, i.e.
\begin{align*}
    C = \ker (H) = \lbrace c \in \mathbb{F}_q^n \, | \, c\cdot H^\top = 0 \rbrace .
\end{align*}

A  matrix $A\in\Fq^{r\times s}$ is said to have \textit{row-weight} $w$, for some nonnegative integer $w$, if every row of $A$ has Hamming weight equal to $w$. Similarly, we say that $A$ has \emph{column-weight} $v$, if each of its columns has Hamming weight $v$.\\

In the following we will focus on the family of
moderate density parity-check (MDPC) codes. They are an extension of the well-known low density parity-check (LDPC) codes, 
and they are defined by the row-weight of a parity-check matrix. The terminology was first introduced in \cite{ou09u}, and then these codes were reintroduced and further generalized in \cite{mi13} for cryptographic purposes.

\begin{definition}
    Let $\{C_i\}$ be a family of binary linear codes of length $n_i$ with parity-check matrix $H_i$. If $H_i$ has row weight $\mathcal{O}(\sqrt{n_i})$, $\{C_i\}$ is called a (family of) \textit{moderate density parity-check (MDPC)} code. If, in addition, the weight of every column of $H_i$ is a constant $v_i$ and the weight of every row of the $H_i$ is a constant $w_i$ we say the MDPC code is of type $(v_i, w_i)$.
\end{definition}

MDPC codes have been constructed in various ways. In their seminal paper \cite{ou09u},  Ouzan and  Be'ery designed cyclic MDPC codes carefully choosing the idempotent generator of the dual code. This structure has been generalized in order to design quasi-cyclic MDPC codes (see e.g. \cite{Janoska,mi13}). 
A different approach has been proposed in \cite{Tillich}, where a random model is considered.

In the definition of an MDPC code the chosen parity-check matrix is very important. Indeed, as for LDPC codes, an MDPC code automatically comes together with a decoding algorithm -- for instance the bit-flipping algorithm -- whose performance depends on the chosen parity-check matrix.
Thus, in order to study the error-correction performance, we introduce the following quantity.

\begin{definition}
    Let $H$ be a binary matrix. The \textit{maximum column intersection} is the maximal cardinality of the intersection of the supports of any pair of distinct columns of $H$.
\end{definition}

The following result was found by Tillich in 2018 (for more details and the proof see \cite{Tillich}). It states the amount of errors that can be corrected within one round of the bit-flipping decoding algorithm.

\begin{theorem}\label{Tillich_correction}
    Let $C$ be an MDPC code of type $(v, w)$ with parity-check matrix $H$. Let $s_H$ denote the maximum column intersection of $H$. Performing one round of the bit-flipping decoding algorithm with respect to $H$, we can correct all errors of weight at most $\lfloor \frac{v}{2 \cdot s_H}\rfloor$.
\end{theorem}
\noindent It hence follows that, the smaller $s_H$, the more errors can be corrected after one round of the bit-flipping decoding algorithm. A random construction would yield an asymptotic value for $s_H$. We would like to design MDPC codes in such a way that $s_H$ is as small as possible and, more importantly, that $s_H$ is deterministic.

\section{MDPC codes from Projective Planes}\label{sec:planes}

The projective plane $\PG(2,q)$ is a point-line geometry constructed from a three-dimensional vector space $V$ over $\Fq$. Its points and lines are the one- and two-dimensional subspaces of $V$ respectively and the containment relation in $V$ defines the incidence relation of the plane. It has $q^2+q+1$ points and equally many lines. The geometry satisfies the following properties:

\begin{enumerate}
        \item any two distinct points are incident with exactly one common line;
        \item any two distinct lines are incident with exactly one common point;
        \item there are four points such that no three of them are collinear.
    \end{enumerate}
    
This means that $\PG(2,q)$ can also be regarded as a symmetric $2$-$(q^2+q+1,q+1,1)$-design, where the lines correspond to the blocks.
Moreover, every line in $\PG(2,q)$ is incident with $q+1$ points and dually, every point is incident with $q+1$ lines.
One way to represent $\PG(2,q)$ is by an incidence matrix. This is a matrix $A$ whose rows and columns are indexed by points and lines respectively such that

\begin{align*}
    (A)_{p\ell} = \begin{cases}
               1 \text{ if } p \text{ is incident with } \ell \\
               0 \text{ otherwise.} 
            \end{cases}
\end{align*}

Here we describe an alternative way to represent the projective plane $\PG(2,q)$. We can identify the set of points with the integers modulo $q^2+q+1$. For the description of the lines, we will follow the instruction presented by Hirschfeld in \cite[p77---p79]{Hirsch}. Let us therefore introduce the following set. 

\begin{definition}\label{def:perfect_diff_set}
A set $D=\{d_0, \dots , d_r\}\subseteq \mathbb Z/(r^2+r+1)\mathbb Z$  is called a \textit{perfect difference set}, if all differences $(d_i - d_j)$ are distinct modulo $r^2+r+1$, for $i, j \in \{0, \dots, r\}$.
\end{definition}

\begin{example}
For instance, consider $r=2$. One can show that the set $D = \{0, 1, 3\}$ of $r+1 = 3$ integers is indeed a perfect difference set, since any two differences between two distinct elements are pairwise disjoint modulo $r^2+r+1 = 7$.
\end{example}

Hirschfeld showed in \cite[Theorem 4.2.2 and its Corollary]{Hirsch} that the set of lines of $\PG (2, q)$ is fully described by the circulant shifts modulo $q^2+q+1$ of a perfect difference set of $q+1$ elements. In this way we obtain a circulant incidence matrix in which the support of the first column is $D$.\\
In order to illustrate this, consider the Fano plane $\PG (2, 2)$ consisting of seven points and lines. We have seen, that the points will be identified with the integers modulo $q^2+q+1 = 7$. For the set of lines we will use the cyclic shifts (modulo 7) of the set $D = \{ 0, 1, 3 \}$, which we have seen is in fact a perfect difference set. Explicitly, we obtain the following set of points $\mathcal{P}$ and set of lines $\mathcal{L}$
\begin{align*}
    \mathcal{P} &= \{ 0, 1, 2, 3, 4, 5, 6 \}, \\
    \mathcal{L} &= \lbrace \{ 0+i, 1+i, 3+i \} \, | \, i \in \{0, \ldots 6 \} \rbrace.
\end{align*}

The defining properties of projective planes have made them a good source of error-correcting codes by taking their incidence matrices as the parity-check matrix, as was done already in the late 1950s, cf. \cite{Prange} or \cite{Rudolph}.
\begin{definition}
Let $H$ be an incidence matrix of $\Uppi = \PG(2,q)$ over the binary finite field $\mathbb{F}_2$. We define the code in $\F_2^{q^2+q+1}$
\begin{align*}
    C_2(\Uppi)^{\perp} =\ker (H).
\end{align*}
\end{definition}

Codes from planes have been intensively studied and many properties have been derived thanks to the underlying geometric structure. Among the most relevant properties, Graham and MacWilliams \cite{Graham} completely determined the dimension of the codes $C_p(\Uppi)^\perp$ over $\F_p$ and their minimum distance when $p=2$ was determined by Assmus and Key \cite{Assmus}. Here we state the two results, restricting ourselves only to the case $p=2$.

\begin{theorem}\label{thm:parameters_code_plane}
    The code $C_2(\Uppi)^\perp$ is a $[q^2+q+1,k,d]_2$ code, where
    $$ (k,d)=\begin{cases} (1,q^2+q+1) & \mbox{ if  } q \mbox{ is odd }, \\
    (2^{2h}-3^h+2^h, 2^h+2) & \mbox{ if } q=2^h.  
    \end{cases}$$
\end{theorem}

The first part just follows from the observation that if $A$ is the incidence matrix of a projective plane of order $q$, then by definition

\[A^\top \!A = AA^\top = qI + J,\]
where $I$ is the identity matrix and $J$ the all-one matrix of size $q^2+q+1$.

From Theorem \ref{thm:parameters_code_plane} we can see that binary codes from $\PG(2,q)$ are only interesting whenever $q$ is even. Moreover, one can see that the incidence matrix of $\Uppi$ has constant row and column weight equal to $q+1$ which is $\mathcal O(\sqrt{q^2+q+1})$. Hence, codes from projective planes are very special examples of MDPC codes. With the aid of Theorem \ref{Tillich_correction}, we can show that one round of the bit-flipping algorithm on these codes permits to decode up to half the minimum distance with no failure probability, for any projective plane.

\begin{theorem}\label{thm:correction_capability_plane_codes}
       Let $\Uppi$ be a projective plane of even order and $H$ its incidence matrix, which is the parity-check matrix of the code $C_2(\Uppi)^\perp$.  After performing one round of bit-flipping on $H$ we can correct any error of weight up to $\lfloor\frac{d-1}{2}\rfloor$, where $d$ is the minimum distance of $C_2(\Uppi)^\perp$.
\end{theorem}

\begin{proof}
     Since a projective plane is in particular a symmetric $2$-$(q^2+q+1,q+1,1)$-design, then the maximum column intersection of $H$ is $1$. Moreover, the matrix $H$ is of type $(q+1,q+1)$. Hence, applying Theorem \ref{Tillich_correction}, we obtain that one round of the bit-flipping algorithm corrects every error of weight at most $\lfloor\frac{d-1}{2}\rfloor$.
\end{proof}

Theorem \ref{thm:correction_capability_plane_codes} shows that codes from planes are really powerful, and have the best performance according to Theorem \ref{Tillich_correction}, for a given matrix  of type $(q+1,q+1)$ and size $(q^2+q+1)\times(q^2+q+1)$. However, we can only construct codes from projective planes of even order, resulting in  $[2^{2h}+2^h+1,2^{2h}-3^h+2^h, 2^h+2]_2$ codes. 
This lack of choice of the parameters motivated many variation on this construction. In the last 50 years, many codes have been constructed based on underlying geometric objects: Euclidean and projective geometries over finite fields \cite{delsarte1969geometric,tang2005codes,ko01}, linear representation of Desarguesian projective planes \cite{pepe2009small}, (semi-)partial geometries \cite{johnson2004codes,Vandendriessche2010}, generalized quadrangles \cite{vontobel2001construction,kim2007small}, generalized polygons \cite{liu2005ldpc}, Ramanujan graphs \cite{Margulis,ro00p}, $q$-regular bipartite graphs from point line geometries \cite{Kim04}  and other incidence structures coming from combinatorial designs \cite{johnson2001regular,johnson2001construction,weller2003regular,johnson2004low}.

For the same reason, we propose a new construction of (families of) MDPC codes based on a suitable system of conics in a Desarguesian projective plane that behaves itself like a projective plane. This is encapsulated in the concept of \emph{projective bundles}, which we define in the following section.

\section{MDPC codes from Projective Bundles}\label{section_construction}

In this section we present the new MDPC codes using projective bundles by constructing its parity-check matrix. We start off by introducing the relevant geometrical objects, which are ovals and projective bundles in $\PG(2,q)$.

\begin{definition} 
	An oval in $\PG(2,q)$ is a set of $q+1$ points, such that every line intersects it in at most two points.
\end{definition}

The classical example of an oval is a non-degenerate conic, i.e. the locus of an irreducible homogeneous quadratic equation. When $q$ is odd, Segre's seminal result \cite{Segre2} shows that the converse is also true: every oval is a conic.

\begin{definition}
    A line in $\PG(2,q)$ is \textit{skew}, \textit{tangent} or \textit{secant} to a given oval if it intersects it in zero, one or two points respectively.
\end{definition}

We recall some properties of ovals which were first recorded by Qvist \cite{Qvist}. We include the proof as it will be relevant later.

\begin{lemma}\label{lem:ovaltangents}
    An oval in $\PG(2,q)$ has $q+1$ tangent lines, one in each point. 
    \begin{itemize} 
        \item If $q$ is odd, every point not on the oval is incident with zero or two tangent lines.
        \item If $q$ is even, then all tangent lines are concurrent.
    \end{itemize}
\end{lemma}

\begin{proof}
    Consider a point on the oval. Then there are $q$ lines through this point intersecting the oval in one more point. This means that one line remains, which is necessarily a tangent line, hence proving the first part of the lemma.
    
    Now suppose that $q$ is odd and consider a point on a tangent line, not on the oval. As $q+1$ is even, this point is incident with an odd number of tangent lines more. Since the point is arbitrary, and there are $q+1$ tangent lines, this implies that every point on the tangent line (but not on the oval) is incident with exactly two tangent lines.    
    
    When $q$ is even, we consider a point on a secant line, but not on the oval and proceed in a similar fashion as before: $q-1$ is odd, so this point is incident with an odd number of tangents. Since this point is arbitrary, and there are $q+1$ tangent lines, this implies that every point on the secant line is incident with exactly one tangent line. Therefore the intersection point of two tangent lines is necessarily the intersection of all tangent lines.
\end{proof}

When $q$ is even, one can add the point of concurrency of the tangent lines, which is called the \textit{nucleus}, to the oval to obtain a set of $q+2$ points that has zero or two points in common with every line. This leads us to the following definition.

\begin{definition}
    A \textit{hyperoval} is set of $q+2$ points in $\PG(2,q)$ such that every line has zero or two points in common. A \textit{dual hyperoval} is a set of $q+2$ lines such that every point is incident with zero or two lines.
\end{definition}

We will encounter these objects again later on. We are now in the position to define projective bundles.

\begin{definition}
    A \textit{projective bundle} is a collection of $q^2+q+1$ ovals of $\PG(2,q)$ mutually intersecting in a unique point.
\end{definition}

Projective bundles were introduced by Glynn in his Ph.D.\ thesis \cite{Glynn} under the name `packings of $(q+1)$-arcs'. The original definition is a bit more general and applies to any projective plane instead of just $\PG(2,q)$. Since the only known projective bundles exist in $\PG(2,q)$, it suffices for our purposes to restrict ourselves to this case.

It follows from the definition that one can consider the points of $\PG(2,q)$ and the ovals of a projective bundle as the points and lines of a projective plane of order $q$. We can then define the notion of secant, tangent and skew ovals (which belong to the projective bundle) with respect to a line. Moreover, one can interchange the role of lines and ovals in the proof of \ref{lem:ovaltangents} and find the following statement, which we record for convenience.

\begin{lemma}\label{lem:linetangents}
    Given a projective bundle, a line in $\PG(2,q)$ has $q+1$ tangent ovals, one in each point. 
    \begin{itemize} 
        \item If $q$ is even, then all tangent ovals are concurrent.
        \item If $q$ is odd, every point not on the line is incident with zero or two tangent ovals.
    \end{itemize}
\end{lemma}

When $q$ is even, we can similarly as before define a \textit{hyperoval of ovals} as a set of $q+2$ ovals such that every point is contained in zero or two of them.

An interesting property of projective bundles is that a third projective plane can be found. This result is due to Glynn \cite[Theorem 1.1.1]{Glynn} and served as the motivation for projective bundles: to possibly find new projective planes from known ones.

\begin{theorem}\label{cor:tangencyprojplane}
    Consider the ovals of a projective bundle and the lines of $\PG(2,q)$ as points and lines respectively, with incidence defined by tangency. Then this point-line geometry is a projective plane of order $q$. 
\end{theorem}

We can rephrase this in terms of incidence matrices. As follows: if $A$ and $B$ are the point-line incidence matrices of the $\PG(2,q)$ and the projective plane whose lines are the ovals of a projective bundle, then $AB^\top \!\pmod{2}$ is again the incidence matrix of a projective plane. However, for $q$ even this idea to construct new projective planes does not work, since then all three projective planes are isomorphic \cite[Corollary 1.1.1]{Glynn}.

Glynn showed that projective bundles indeed exist for any $q$, and his examples are all bundles of conics. When $q$ is odd, he showed the existence of three distinct types of projective bundles in $\PG(2,q)$, by identifying them with planes in $\PG(5,q)$. It was shown in \cite{Baker} that perfect difference sets can also be used to describe these projective bundles. In fact, given a perfect difference set $D\subseteq \Z/(q^2+q+1)\Z$ and its circular shifts corresponding to the set of lines of $\PG (2, q)$, the three bundles are represented in the following way.
\begin{enumerate}
    \item \textit{Cirumscribed bundle:} set of all circular shifts of $-D$.
    \item \textit{Inscribed bundle:} set of all circular shifts of $2D$.
    \item \textit{Self-polar bundle:} set of all circular shifts of $D/2$.
\end{enumerate}

We are now going to construct the parity-check matrix as mentioned at the beginning of this section. Let us denote the projective plane formed by the points and lines of $\PG(2,q)$ by $\Uppi$ and the one formed by the points and the ovals of a projective bundle of $\PG(2,q)$ by $\Upgamma$. Then define
\begin{equation}\label{eq:parity_check_code_bundle}
    H = ( \,A \, \mid \, B\,),
\end{equation}
where $A$ and $B$ are the incidence matrices of $\Uppi$ and $\Upgamma$ respectively.
Hence, we obtain a $(q^2+q+1)\times 2(q^2+q+1)$ binary matrix defined by the points, lines and ovals of a projective bundle of $\PG(2,q)$.

\begin{definition}
	A binary linear code with parity-check matrix $H$ given in \eqref{eq:parity_check_code_bundle} is called a \emph{projective bundle code} and we will denote it by
	\begin{align*}
    	C_2(\Uppi\sqcup\Upgamma)^\perp = \ker (H).
	\end{align*}
\end{definition}

Clearly, the matrix $H$ given in \eqref{eq:parity_check_code_bundle} has constant row-weight $w = 2(q+1)$ and constant column-weight $v = q+1$. Hence, $C_2(\Uppi\sqcup\Upgamma)^\perp$ is an MDPC code of length $n=2(q^2+q+1)$ and type $(q+1, 2(q+1))$.\\
\begin{remark}
The family of MDPC codes that we are considering is built upon a parity-check matrix as in \eqref{eq:parity_check_code_bundle}. In such a matrix the number of columns is twice the number of rows and this  coincides with the setting  originally studied in \cite{mi13}.
\end{remark}

\begin{example}
Let us give a short example of a projective bundle code for a relatively small parameter $q = 3$. Hence, we consider the projective plane $\PG (2, 3)$. Recall, that the set of points $\mathcal{P}$ is given by the set of integers modulo $q^2+q+1 = 13$. The set of lines $\mathcal{L}$ is defined by the image of a perfect difference set $D$ of four integers under repeated application of the Singer cycle $S(i) = i+1$. It is easy to verify that $D = \{ 0, 1, 3, 9 \}$ is a perfect difference set, i.e. $$\mathcal{L} = \lbrace \{ 0+i, 1+i, 3+i, 9 + i\} \, | \, i \in \Z/13\Z \rbrace.$$
At this point, let us choose an inscribed bundle $\mathcal{B}_I$ in $\PG (2, 3)$. As shown above, this bundle is represented by the cyclic shifts of $2D = \lbrace 0, 2, 5, 6 \rbrace$. Hence, we obtain
$$\mathcal{B}_I = \lbrace \{ 0+i, 2+i, 5+i, 6 + i\} \, | \, i \in \Z/13\Z \rbrace.$$
Concatenating the two corresponding incidence matrices $A$ and $B$ yields the desired parity-check matrix
\begin{small}
\begin{align*}
H = \left(\begin{array}{ccccccccccccc|ccccccccccccc}
1 & \cdot & \cdot & \cdot & 1 & \cdot & \cdot & \cdot & \cdot & \cdot & 1 & \cdot & 1 & 1 & \cdot & \cdot & \cdot & \cdot & \cdot & \cdot & 1 & 1 & \cdot & \cdot & 1 & \cdot \\
1 & 1 & \cdot & \cdot & \cdot & 1 & \cdot & \cdot & \cdot & \cdot & \cdot & 1 & \cdot & \cdot & 1 & \cdot & \cdot & \cdot & \cdot & \cdot & \cdot & 1 & 1 & \cdot & \cdot & 1 \\
\cdot & 1 & 1 & \cdot & \cdot & \cdot & 1 & \cdot & \cdot & \cdot & \cdot & \cdot & 1 & 1 & \cdot & 1 & \cdot & \cdot & \cdot & \cdot & \cdot & \cdot & 1 & 1 & \cdot & \cdot \\
1 & \cdot & 1 & 1 & \cdot & \cdot & \cdot & 1 & \cdot & \cdot & \cdot & \cdot & \cdot & \cdot & 1 & \cdot & 1 & \cdot & \cdot & \cdot & \cdot & \cdot & \cdot & 1 & 1 & \cdot \\
\cdot & 1 & \cdot & 1 & 1 & \cdot & \cdot & \cdot & 1 & \cdot & \cdot & \cdot & \cdot & \cdot & \cdot & 1 & \cdot & 1 & \cdot & \cdot & \cdot & \cdot & \cdot & \cdot & 1 & 1 \\
\cdot & \cdot & 1 & \cdot & 1 & 1 & \cdot & \cdot & \cdot & 1 & \cdot & \cdot & \cdot & 1 & \cdot & \cdot & 1 & \cdot & 1 & \cdot & \cdot & \cdot & \cdot & \cdot & \cdot & 1 \\
\cdot & \cdot & \cdot & 1 & \cdot & 1 & 1 & \cdot & \cdot & \cdot & 1 & \cdot & \cdot & 1 & 1 & \cdot & \cdot & 1 & \cdot & 1 & \cdot & \cdot & \cdot & \cdot & \cdot & \cdot \\
\cdot & \cdot & \cdot & \cdot & 1 & \cdot & 1 & 1 & \cdot & \cdot & \cdot & 1 & \cdot & \cdot & 1 & 1 & \cdot & \cdot & 1 & \cdot & 1 & \cdot & \cdot & \cdot & \cdot & \cdot \\
\cdot & \cdot & \cdot & \cdot & \cdot & 1 & \cdot & 1 & 1 & \cdot & \cdot & \cdot & 1 & \cdot & \cdot & 1 & 1 & \cdot & \cdot & 1 & \cdot & 1 & \cdot & \cdot & \cdot & \cdot \\
1 & \cdot & \cdot & \cdot & \cdot & \cdot & 1 & \cdot & 1 & 1 & \cdot & \cdot & \cdot & \cdot & \cdot & \cdot & 1 & 1 & \cdot & \cdot & 1 & \cdot & 1 & \cdot & \cdot & \cdot \\
\cdot & 1 & \cdot & \cdot & \cdot & \cdot & \cdot & 1 & \cdot & 1 & 1 & \cdot & \cdot & \cdot & \cdot & \cdot & \cdot & 1 & 1 & \cdot & \cdot & 1 & \cdot & 1 & \cdot & \cdot \\
\cdot & \cdot & 1 & \cdot & \cdot & \cdot & \cdot & \cdot & 1 & \cdot & 1 & 1 & \cdot & \cdot & \cdot & \cdot & \cdot & \cdot & 1 & 1 & \cdot & \cdot & 1 & \cdot & 1 & \cdot \\
\cdot & \cdot & \cdot & 1 & \cdot & \cdot & \cdot & \cdot & \cdot & 1 & \cdot & 1 & 1 & \cdot & \cdot & \cdot & \cdot & \cdot & \cdot & 1 & 1 & \cdot & \cdot & 1 & \cdot & 1 \\
\end{array}\right) ,
\end{align*}
\end{small}
where the zero entries in the parity-check matrix are represented by dots.
\end{example}

\begin{remark}
 Observe that the matrix $H$ defined in \eqref{eq:parity_check_code_bundle} can be constructed from a perfect difference set $D$, by taking the circular shifts of $D$ and $sD$, with $s \in \{-1,2,2^{-1}\}$. Such a matrix has a double circulant structure. Thus, the resulting code $C_2(\Uppi\sqcup \Upgamma)^\perp$ is quasi-cyclic of index $2$, and encoding can be
achieved in linear time and implemented with linear feedback shift registers. Furthermore, we can also deduce -- because of the circular structure -- that the number of bits required to describe 
the parity check matrix is about half the block length. It would even be less if one compresses 
the data.
\end{remark}

In the following subsections we will analyse the dimension, minimum distance and error-correction performance with respect to the bit-flipping decoding algorithm of $C_2(\Uppi\sqcup\Upgamma)^\perp$.

\subsection{Dimension}


Recall from Theorem \ref{thm:parameters_code_plane} that a $p$-ary code $C_p(\Uppi)$ from a projective plane $\Uppi \cong \PG(2,q)$, is either trivial of codimension $1$ -- when $p \nmid q$ -- or it is non-trivial to determine its dimension -- when $p \mid q$.
In our case, the structure of our code allows to both have a non-trivial code and to determine the exact dimension for all $q$. To do so, recall that if $A$ is the incidence matrix of a projective plane of order $q$, then
\[AA^\top = A^\top \!A = qI + J,\]
where $J$ is the all-one matrix of appropriate size.

Using this result we are able to state the dimension of $C_2(\Uppi\sqcup\Upgamma)^\perp$. 

\begin{proposition}\label{prop:dimension_basecase}
    Let $\Uppi$ be a projective plane of order $q$ and let $\Upgamma$ be a projective bundle in $\Uppi$. Then, $$\dim\left(C_2(\Uppi\sqcup\Upgamma)^\perp\right) = \begin{cases}
    q^2+q+2 & \mbox{ if } q \mbox{ is odd, }\\
    2^{2h+1}+2^{h+1}-2(3^h)+1 & \mbox{ if } q=2^h.
    \end{cases}$$
\end{proposition}

\begin{proof}
    In order to determine the dimension of the code, we need to compute the rank of a parity-check matrix $H = ( \,A \, \mid \, B\,)$. Since $H$ is of size $(q^2+q+1)\times 2(q^2+q+1)$, we can already say that the rank of $H$ is at most $q^2+q+1$. Now we consider the two cases.
    
    \noindent \underline{\textbf{Case I: $q$ odd}}. We know from Theorem \ref{thm:parameters_code_plane} that $\rank(A) = q^2+q$, which gives us the lower bound $\rank(H) \geq \rank(A) = q^2+q$.\\
    The matrix $H$ has full rank $q^2+q+1$ if and only if there exists no element in the left-kernel, i.e. if there is no non-zero vector $x\in \mathbb{F}_2^{q^2+q+1}$ such that
    \begin{align}\label{leftkernel_element}
        x  H = 0.
    \end{align}
    However, if $x$ is the all-one vector then Equation (\ref{leftkernel_element}) is satisfied. Hence, there is an element in the cokernel which implies that $H$ cannot have full rank and we conclude that $\dim C_2(\Uppi\sqcup\Upgamma)^\perp = q^2+q+2$. 
    
    \noindent \underline{\textbf{Case II: $q$ even}}. In this case, we consider the matrix 
    $$ H^\top \!H=\begin{pmatrix} A^\top\! A \;& A^\top \!B\\ B^\top \!A \;& B^\top \!B
    \end{pmatrix}=\begin{pmatrix} J & A^\top \!B\\ (A^\top \!B)^\top& J
    \end{pmatrix}.$$
    By Theorem \ref{cor:tangencyprojplane} and the discussion below, $A^\top \!B = C$ is again the incidence matrix of $\PG(2,q)$, and hence the sum of all its rows/columns is equal to the all one vector. Therefore, by doing row operations on $H^\top\!H$, we obtain the matrix
    $$\begin{pmatrix} 0 & A^\top \!B+J\\ (A^\top \!B)^\top& J
    \end{pmatrix}, $$
    which has the same rank as $H^\top \!H$. Hence,
    $$ \rank(H)\geq \rank(H^\top\!H)=\rank(A^\top \!B)+\rank(A^\top \!B+J)\geq 2\rank(A^\top\!B)-1,$$
    where the last inequality comes from the fact that $J$ has rank $1$, and the rank satisfies the triangle inequality. On the other hand, we have that the all one vector is in the column spaces of both $A$ and $B$, showing that $\rank(H)\leq \rank(A)+\rank(B)-1$. Since $A$, $B$ and $A^\top \!B$ are all incidence matrices of a Desarguesian plane, they all have the same rank. Therefore, combining the two inequalities, we obtain
    $$ \rank(H)=2\rank(A)-1,$$
    and using Theorem \ref{thm:parameters_code_plane}, we can conclude that 
    \begin{align*}
        \dim\left(C_2(\Uppi\sqcup\Upgamma)^\perp\right)&=2(q^2+q+1)-\rank(H)=2(q^2+q+1)-2\rank(A)+1\\
        &=2\dim(C_2(\Uppi)^\perp))+1=2^{2h+1}+2^{h+1}-2(3^h)+1. 
    \end{align*}
\end{proof}

\noindent We can thus already say that $C_2(\Uppi\sqcup\Upgamma)^\perp$ is a $[2(q^2+q+1), q^2+q+2]_2$ MDPC code of type $(q+1,2q+1)$.


\subsection{Minimum Distance}

As mentioned earlier, we are interested in the error-correction capability. A relevant quantity to give information about error-correction and also error-detection is the minimum distance of a linear code. 

In the following we will determine the exact value of the minimum distance of $C_2(\Uppi\sqcup\Upgamma)^\perp$. An important observation for the proof is that geometrically, the support of a codeword of $\code$ corresponds to a set of lines and ovals such that every point of $\PG(2,q)$ is covered an even number of times.

\begin{theorem}\label{thm:minimumdistance_codeebundle}
    The minimum distance of $\code$ is $q+2$ and the supports of the minimum weight codewords can be characterized, depending on the parity of $q$. For $q$ odd, the support of a minimum weight codeword is
    \begin{itemize}
        \item an oval and its $q+1$ tangent lines, or
        \item a line and its $q+1$ tangent ovals.
    \end{itemize}
    On the other hand for $q$ even, we find that the support of a minimum weight codeword is
    \begin{itemize}
        \item a dual hyperoval, or
        \item a hyperoval of ovals.
    \end{itemize}
\end{theorem}

\begin{proof}
    Take a codeword of minimum weight in $\code$ and consider its support. This is a set of $r$ lines $L$ and $s$ ovals $O$ such that every point in $\PG(2,q)$ is incident with an even number of these elements. We will show that $r+s \geq q+2$ and equality only holds for the two examples stated. 
    
    Let $a_i$, $0 \leq i \leq 2q+2$, be the number of points that are covered $i$ times, then we can double count the tuples $(P), (P,E_1), (P,E_1,E_2)$, where $P$ is a point and $E_1,E_2 \in L \cup O$ are lines or ovals incident with this point. Remark that by assumption $a_i = 0$ whenever $i$ is odd. We find the following three expressions:
    
    \begin{align}
        \sum_{i=0}^{2q+2} a_i &= q^2+q+1 \\
        \sum_{i=0}^{2q+2} ia_i &= (r+s)(q+1) \\
        \sum_{i=0}^{2q+2} i(i-1)a_i &\leq r(r-1)+s(s-1)+2rs,
    \end{align}
    where the last inequality follows as a line and oval intersect in at most two points. From these equations, we can find $\sum_{i=0}^{2q+2}i(i-2) a_i \leq (r+s)(r+s-q-2)$ and hence $r+s \geq q+2$, as the sum on the left-hand side has only non-negative terms. Moreover, in the case of equality, $a_i = 0$ whenever $i \notin \{0,2\}$.
    
    Now consider a codeword of weight $r+s = q+2$, consisting of $r$ lines $L$ and $s$ ovals $O$. We will investigate the cases $q$ odd and even separately and show the characterisation.
    
    \noindent \underline{\textbf{Case I: $q$ odd}}. Since $q+2$ is odd and hence one of $r$ or $s$ is, we can suppose without loss of generality that $r$ is odd. The argument works the same when $s$ is odd, by interchanging the roles of lines and ovals.
    
    Consider a line not in $L$. Then this line is intersected an odd number of times by the $r$ lines in $L$. Therefore, it should be tangent to an odd number of ovals in $O$, recalling that every point is incident with zero or two elements from $L \cup O$. In particular, any line not in $L$ is tangent to at least one oval in $O$. So count the $N$ pairs $(\ell,c)$, where $\ell$ is a line not in $L$, $c \in O$ and $|\ell \cap O| = 1$. By the previous observation, it follows that $q^2+q+1-r = q^2-1+s \leq N$. On the other hand, a oval has $q+1$ tangent lines so that $N \leq s(q+1)$. Combining these two leads to $s \geq q$, which implies that $r = 1$ and $s = q+1$. Remark that this argument only depends on $r$ being odd.
    
    If $o \in O$ is one of these $q+1$ ovals, we see that the other $q$ ovals intersect $O$ in $q$ distinct points, as no point is incident with more than two elements from $L \cup O$. This immediately implies that the unique line in $L$ must be tangent to $o$. As $O$ was arbitrary, we conclude that the support of the codeword consists of one line and $q+1$ ovals tangent to it. By Lemma \ref{lem:linetangents} this indeed gives rise to a codeword, as every point not on the line is incident with zero or two ovals.
    
    \noindent \underline{\textbf{Case II: $q$ even}}. The situation is slightly different. Since $q+2$ is even now, either $r$ and $s$ are both odd, or both even. When $r$ is odd, we can reuse the argument from before to find the configuration of $q+1$ ovals tangent to a line. However, by Lemma \ref{lem:linetangents} we know that these $q+1$ ovals are all incident with a unique point, which is hence covered $q+1$ times, a contradiction. 
     
    So suppose that $r$ and $s$ are even. Any line in $L$ is intersected by the  $r-1$ other lines in $L$, leaving $q+1-(r-1)$ points to be covered by the ovals in $O$, which is an even number. We see that we must have an even number of tangent ovals to this line. Similarly for a line not in $L$, we observe that it is intersected an even number of times by the $r$ lines in $L$ and hence it should have an even number of intersections with the ovals in $O$, leading again to an even number of tangent ovals. In summary, every line in $\PG(2,q)$ is incident with an even number of tangent ovals. Now, by Lemma \ref{lem:linetangents} and the fact that every point is covered zero or twice by the elements of $O \cup L$, it follows that every line in $\PG(2,q)$ is incident with zero or two tangent ovals. So suppose that $s > 0$, meaning we have at least one oval in $O$ and consider its $q+1$ tangent lines. Then each of these lines should have one more tangent oval, and all of these are distinct by Corollary \ref{cor:tangencyprojplane}, which means we find $s = q+2$ ovals forming a hyperoval of ovals. If $s = 0$, we find a dual hyperoval, concluding the theorem.
\end{proof}

\subsection{Error-Correction Capability}
It is well-known that the minimum distance of a code gives information about the decoding radius. This means that it reveals an upper bound on the amount of errors that can be always detected and  corrected.\\

We would like to focus in this subsection here on the performance of the constructed MDPC code $C_2(\Uppi\sqcup\Upgamma)^\perp$ within one round of the bit-flipping decoding algorithm. We now adapt and apply Theorem \ref{Tillich_correction} to the parity-check matrix $H$ of $C_2(\Uppi\sqcup\Upgamma)^\perp$ given in \eqref{eq:parity_check_code_bundle}. 

\begin{proposition}\label{bit_flipping_bound_mdpc}
The intersection number of the matrix $H$ defined in \eqref{eq:parity_check_code_bundle} is  $s_H=2$. Thus,
  after performing one round of the bit-flipping algorithm on $H$ we can correct all the errors of weight at most  $\lfloor\frac{q+1}{4}\rfloor$ in the code $C_2(\Uppi\sqcup\Upgamma)^\perp$.
\end{proposition}

\begin{proof}
      From the construction of $H$ we have that $H$ consists of two matrices $A$ and $B$ which are the incidence matrices of points and lines and points and ovals of a projective bundle in $\PG(2,q)$, respectively. Clearly, both matrices $A$ and $B$ have a maximum column intersection equal to 1 as two distinct lines in a projective plane intersect in exactly one point and a similar property holds for every pair of distinct ovals of a projective bundle by definition. Since every line intersects an oval in at most 2 points, the maximum column intersection of the matrix $H$ is at most $2$. On the other hand, if we consider any two distinct points on an oval in the projective bundle, there always exists a line passing through them. Hence, $s_H=2$. The second part of the statement then follows directly from Theorem \ref{Tillich_correction}.
\end{proof}

\begin{remark}\label{rem:sh}
Observe that $s_H = 1$ for a parity-check matrix of size $(q^2+q+1)\times c$ and column weight $q+1$ implies $c \leq q^2+q+1$. this can be seen by counting the tuples $\{(x,y,B) \,\, | \,\, x,y \in B\}$ in two ways. Thus, the value $s_H = 2$ is the best possible for $c > q^2+q+1$. 
Furthermore, compared to a random construction of MDPC code, our design guarantees a deterministic error-correction performance for one round of the bit-flipping decoding algorithm. In particular, for the random model proposed in \cite{Tillich} it was proved that the expected value of $s_H$ is $\mathcal O(\frac{\log n}{\log \log n})$. Hence, our construction guarantees an error-correction capability of the bit-flipping algorithm which improves the random construction by a factor $\mathcal O(\frac{\log n}{\log \log n})$.
\end{remark}

Additionally, we have implemented the parity-check matrix for our MDPC-design as well as one round of the bit-flipping decoding algorithm. We were interested if we could correct even more errors than the number guaranteed in Proposition \ref{bit_flipping_bound_mdpc}. Since the bit-flipping decoding algorithm is only dependent on the syndrome and not on the actual chosen codeword, we took the all-zero codeword and added a pseudo-random error-vector of a fixed weight $wt(e) \geq \lfloor \frac{q+1}{4} \rfloor$. We have generated $10^5$ distinct error vectors. Each of these error vectors then was used to run one round of the bit-flipping decoding algorithm for all the three different families of MDPC-codes that we have constructed. It turned out that all the three types of our constructed code showed exactly the same error-correction performance.\\
Finally, we have computed the probability of successful error-correction for the parameters $q \in \lbrace 5, 7, 9, 11, 13, 17, 19, 23, 25 \rbrace $. The following results were obtained for the different error weights.

\begin{table}[H]
\setlength{\tabcolsep}{4mm} 
\def\arraystretch{1.25} 
\centering
\begin{tabular}{|c|p{2.5cm}|p{2.5cm}|p{2.5cm}|}
\hline
$q$ & $\lfloor \frac{q+1}{4} \rfloor + 1$ errors  &  $\lfloor \frac{q+1}{4} \rfloor + 2$ errors & $\lfloor \frac{q+1}{4} \rfloor + 3$ errors\\
\hline
\hline
5 & 50.82\%  & 0.16\% & - \\
7 & 50.10\%  & 0.34\%  & - \\
9 & 79.31 \%  & 3.86\%  & - \\
11 & 43.83\% & 0.19\% & - \\
13 & 90.4\% & 14.4\% & - \\
17 & 97.2\% & 57.8\% & 7.8\% \\
19 & 91.8\% & 42.6\% & 10.7\% \\
23 & 97.86 \% & 77.66\% & 31.3\% \\
25 & 99.87\% & 95.3\% & 71.25\% \\
\hline
\end{tabular}
\caption{Probability to decode a received word, with error-weight $\lfloor \frac{q+1}{4} \rfloor + i$ for $i=1, 2, 3$, correctly after one round of the bit-flipping decoding algorithm.}\label{table:experiment}
\end{table}
Table \ref{table:experiment} shows that the probability to correct even more errors grows as we increase $q$. This is due to the fact, that for small $q$ we reach the unique decoding radius much faster.

\begin{remark}
In \cite{Tillich} the author analyzed also the error-correction performance after two rounds of the bit-flipping decoding algorithm. More precisely he estimated the probability that one round of the algorithm corrects enough errors so that in the second round all remaining errors will be correctable.
Following the notation of that paper, let us denote by $S$ the number of errors left after one round of the bit-flipping algorithm. Assuming that we have an MDPC code of length $n$ and of type $(v, w)$, where both $v$ and $w$ are of order $\Theta(\sqrt{n})$, the probability that $S$ is at least a certain value $t'$ satisfies the following inequality:
\begin{align}\label{prob:two_round_BF}
    \mathbb{P}\left( S \geq t' \right) \leq \frac{1}{\sqrt{t'}} {\rm e}^{\frac{t' v}{4} \ln(1 - {\rm e}^{-\frac{4wt}{n}}) + \frac{t'}{8}\ln(n) + O\left(t'\ln(t'/t) \right)},
\end{align}
where $t = \Theta(\sqrt{n})$ is the initial amount of errors that were introduced.\\
We have seen in Proposition \ref{bit_flipping_bound_mdpc}, that performing one round of the bit-flipping algorithm to a parity-check matrix $H$ of $C_2(\Uppi \sqcup \Upgamma)^\perp$ we can correct $\lfloor \frac{q+1}{4} \rfloor$ errors. Therefore, a second round of the bit-flipping is able to correct completely if after one round there are no more than $\lfloor \frac{q+1}{4} \rfloor$ errors left. Applying \eqref{prob:two_round_BF} for $t' = \lfloor \frac{q+1}{4} \rfloor$ to the parity-check matrix $H$ of $C_2(\Uppi \sqcup \Upgamma)^\perp$ given in \eqref{eq:parity_check_code_bundle}, we obtain that we can successfully correct every error of weight $t=\Theta(\sqrt{n})$ after two rounds of the bit-flipping decoding algorithm with probability
$e^{-\Omega(n)}.$
\end{remark}

\section{Generalizations}\label{section:generalization}
Since our aim is to have more flexibility in the parameters, here we generalize the approach of Section \ref{section_construction}, by considering several disjoint projective bundles.

Let $t>1$ be a positive integer and let us fix a Desarguesian projective plane $\Uppi=\PG(2,q)$. Let $\Upgamma_1,\ldots, \Upgamma_{t}$ be disjoint\footnote{Here with ``disjoint'' we mean that any two distinct projective bundles have no common oval.} projective bundles of conics in $\Uppi$. Since we want $s_H$ to be low, we cannot take projective bundles of ovals in general, as for example two ovals in $\PG(2,q)$, $q$ even, could intersect in up to $q$ points: take any oval, add the nucleus and delete another point to find a second oval intersecting it in $q$ points. In Proposition \ref{cor:bit-flipping_performance_general} we will see that by choosing conics, we find $s_H = 4$.

Let us denote by $A$ the incidence matrix of $\Uppi$ and by $B_i$ the incidence matrix of the projective bundle $\Upgamma_i$, for each $i\in\{1,\ldots,t\}$. We then glue together all these matrices and consider the code $C_2(\Uppi\sqcup\Upgamma_1\sqcup\ldots\sqcup\Upgamma_t)^\perp$ to be the binary linear code whose parity-check matrix is 
\begin{equation}\label{eq:paritycheck_general}
    H_{q,t}:=\left(
    A \mid B_1 \mid \cdots \mid B_t
    \right).
\end{equation}

As already discussed, it is important to specify which parity-check matrix of a code we consider when we study the decoding properties, since the bit-flipping algorithm depends on the choice of the parity-check matrix.

We focus now on the parameters on the constructed codes. We first start with a result on the dimension of the code $C_2(\Uppi\sqcup\Upgamma_1\sqcup\ldots\sqcup\Upgamma_t)^\perp$

\begin{proposition}
     Let $\Uppi=\PG(2,q)$ be a Desarguesian projective plane of order $q$ and let 
     $\Upgamma_1,\ldots,\Upgamma_{t}$ be a projective bundles in $\Uppi$. Then, 
    \begin{align*}\dim\left(C_2(\Uppi\sqcup\Upgamma_1\sqcup\ldots\sqcup\Upgamma_t)^\perp\right) &= 
    t(q^2+q+1)+1 & \mbox{ if } q \mbox{ is odd, }\\
   \dim\left(C_2(\Uppi\sqcup\Upgamma_1\sqcup\ldots\sqcup\Upgamma_t)^\perp\right) & \geq (t+1)(2^{2h}-3^h+2^h)+t  & \mbox{ if } q=2^h.
    \end{align*}
\end{proposition}

\begin{proof}
The proof goes as for Proposition \ref{prop:dimension_basecase}.

 \noindent \underline{\textbf{Case I: $q$ odd}}. We know from Theorem \ref{thm:parameters_code_plane} that $\rank(A) = q^2+q$, which gives us the lower bound $\rank(H_{q,t}) \geq \rank(A) = q^2+q$.\\
    On the other hand, since $A$ and each matrix $B_i$ has the all one vector in its left kernel, we have that also $H_{q,t}$ has a nontrivial left kernel, and hance $\rank(H_{q,t})=q^2+q$, yielding 
    $$\dim\left(C_2(\Uppi\sqcup\Upgamma_1\sqcup\ldots\sqcup\Upgamma_t)^\perp\right) =(t+1)(q^2+q+1)-\rank(H_{q,t})= t(q^2+q+1)+1.$$
    
    \noindent \underline{\textbf{Case I: $q$ even}}. Let us write $q=2^h$. In this case, we have that the all one vector belongs to the column spaces of each matrix $B_i$. Therefore, 
    $$\rank(H_{q,t})\leq \rank(A)+\sum_{i=1}^t\rank(B_i)-t.$$ 
     Thus, we obtain
    \begin{align*}\dim\left(C_2(\Uppi\sqcup\Upgamma_1\sqcup\ldots\sqcup\Upgamma_t)^\perp\right)&=(t+1)(q^2+q+1)-\rank(H_{q,t}) \\ & \geq (t+1)(q^2+q+1)-\rank(A)-\sum_{i=1}^t\rank(B_i)+t \\
    &=(t+1)\dim(C_2(\Uppi)^\perp)+t\\ &=(t+1)(2^{2h}-3^h+2^h)+t ,\end{align*}
    where the last equality comes from Theorem \ref{thm:parameters_code_plane}.
\end{proof}
Also in this general case we can study the minimum distance of the code $C_2(\Uppi\sqcup\Upgamma_1\sqcup\ldots\sqcup\Upgamma_t)^\perp$, generalizing the result on the minimum distance obtained when $t=1$ in Theorem \ref{thm:minimumdistance_codeebundle}. However, this time we are only able to give a lower bound.

\begin{proposition}\label{prop:minimumdistance_general}
    The minimum distance of $C_2(\Uppi\sqcup\Upgamma_1\sqcup\ldots\sqcup\Upgamma_t)^\perp$ is at least $\big\lceil\frac{q+2}{2}\big\rceil$.
\end{proposition}

\begin{proof}
The proof goes in a similar way as the one of Theorem \ref{thm:minimumdistance_codeebundle}. 
 Take a codeword of minimum weight in $C_2(\Uppi\sqcup\Upgamma_1\sqcup\ldots\sqcup\Upgamma_t)^\perp$ and consider its support. This is a set $L$ of $r$ lines  and a set $O_i$ of $s_i$ ovals for each $i \in\{1,\ldots,t\}$ such that every point in $\PG(2,q)$ is incident with an even number of these elements. We will show that $r+s_1+\ldots+s_t \geq \frac{q+2}{2}$.
 
    Let $a_i$, $0 \leq i \leq 2q+2$, be the number of points that are covered $i$ times, then we can double count the tuples $(P), (P,E_1), (P,E_1,E_2)$, where $P$ is a point and $E_1,E_2 \in L \cup O_1\cup\ldots\cup O_t$ are lines or ovals incident with this point. Remark that by assumption $a_i = 0$ whenever $i$ is odd. We find the following three expressions:
    
    \begin{align}
        \sum_{i=0}^{2q+2} a_i &= q^2+q+1  \\
        \sum_{i=0}^{2q+2} ia_i &= \bigg(r+\sum_{i=1}^ts_i\bigg)(q+1) \label{eq:general2}\\
        \sum_{i=0}^{2q+2} i(i-1)a_i &\leq r(r-1)+\sum_{i=1}^ts_i(s_i-1)+2r\bigg(\sum_{i=1}^ts_i\bigg)+4\bigg(\sum_{1\leq i<j\leq t}s_is_j\bigg), \label{eq:general3}
    \end{align}
    as two conics intersect in at most $4$ points by B\'ezout's theorem.
    Subtracting \eqref{eq:general2} from \eqref{eq:general3} we obtain
    $$0 \leq \sum_{i=0}^{2q+2} i(i-2)a_i = \bigg(r+\sum_{i=1}^ts_i\bigg)\bigg(r-q-2+\sum_{i=1}^ts_i\bigg)+2\bigg(\sum_{1\leq i<j\leq t}s_is_j\bigg).$$
    One can easily check that this  last quantity is in turn at most
    $$ \bigg(r+\sum_{i=1}^ts_i\bigg)\bigg(2r-q-2+2\sum_{i=1}^ts_i\bigg),$$
    which then implies
    $$ r+\sum_{i=1}^ts_i \geq \frac{q+2}{2}.$$
\end{proof}

As a direct consequence of Proposition \ref{prop:minimumdistance_general} we have that in principle it should be possible to correct at least $\big\lfloor \frac{q}{4}  \big\rfloor$ errors in the code $C_2(\Uppi\sqcup\Upgamma_1\sqcup\ldots\sqcup\Upgamma_t)^\perp$ when $q$ is even, and at least $\big\lfloor \frac{q+1}{4}\big\rfloor$ when $q$ is odd. However, also in this case, when running one round of the bit-flipping algorithm on the matrix $H_{q,t}$ given in \eqref{eq:paritycheck_general}, we only correct a smaller fraction of them, as the following result shows.

\begin{proposition}\label{cor:bit-flipping_performance_general}
The intersection number of the matrix $H_{q,t}$ defined in \eqref{eq:paritycheck_general} is at most $4$. Thus,
  after performing one round of the bit-flipping algorithm on $H_{q,t}$ we can correct all the errors of weight at most  $\lfloor\frac{q+1}{8}\rfloor$ in the code $C_2(\Uppi\sqcup\Upgamma_1\sqcup\ldots\sqcup\Upgamma_t)^\perp$.
\end{proposition}

\begin{proof}
  The maximum column intersection is given by the maximum number of points lying in the intersection of elements in $\Uppi\sqcup\Upgamma_1\sqcup\ldots\sqcup\Upgamma_t$. Each pair of lines intersects in exactly a point, and the same holds for every pair of conics belonging to the same projective bundle, since each projective bundle is itself (ismorphic to) a projective plane. Moreover, every line intersects a conic in at most two points, and we have already seen that each pair of conics meets in at most $4$ points. Hence, the maximum column intersection of $H_{q,t}$ is at most $4$. The second part of the statement directly follows from Theorem \ref{Tillich_correction}.
\end{proof}

\begin{remark}
 At this point it is natural to ask whether it is possible to construct disjoint projective bundles, and -- if so -- how many of them we can have. It is shown in \cite[Theorem 2.2]{Baker} that one can always find $(q-1)$ disjoint projective bundles when $q$ is even, and $\frac{q^2(q-1)}{2}$ of them when $q$ is odd. We want to remark that this is not a restriction, since we still want that our codes
 $ C_2(\Uppi\sqcup\Upgamma_1\sqcup\ldots\sqcup\Upgamma_t)^\perp$ (together with the parity-check matrices $H_{q,t}$ of the form \eqref{eq:paritycheck_general}) give rise to  a family of MDPC codes. Thus, we are typically  interested in family of codes  where $t$ is a constant and does not grow with $q$.
\end{remark}

\begin{remark}
This construction provides a better performance of (one round of) the bit-flipping algorithm compared to the one run on random constructions of MDPC codes explained in \cite{Tillich}. As already explained in Remark \ref{rem:sh},  the random construction of MDPC codes provides in average MDPC codes whose maximum column intersection is $\mathcal O(\frac{\log n}{\log \log n})$, and thus one round of bit-flipping algorithm corrects errors of weight at most $\mathcal O(\frac{\sqrt{n}\log\log n}{\log n})$ in these random codes. Hence, also the generalized constructions of codes from projective bundles have asymptotically better performance in terms of the bit-flipping algorithm.
\end{remark}

\section{Conclusion}
In this paper we proposed a new construction of a family of moderate density parity-check codes arising from geometric objects. Starting from a Desarguesian projective plane $\Uppi$ of order $q$ and a projective bundle $\Upgamma$ in $\Uppi$, we  constructed a binary linear code whose parity-check matrix is obtained by concatenating the incidence matrices of $\Uppi$ and $\Upgamma$. We observed that we can construct these two matrices  taking the circular shifts of two perfect difference sets modulo $(q^2+q+1)$, providing a natural structure as a quasi-cyclic code of index $2$. Hence, the storage complexity is linear in the length and the encoding can be achieved in linear time using  linear feedback shift registers. Furthermore, the underlying geometry  of $\Upgamma$ and $\Uppi$ allowed us to study the metric properties of the corresponding code, and we could determine its exact  dimension and  minimum distance. We then analyzed the performance of the bit-flipping algorithm showing that it outperforms asymptotically the one of the random construction of codes obtained  in \cite{Tillich}. 
We then generalized the construction of this family of codes by concatenating the incidence matrices of several disjoint projective bundles living in the Desarguesian projective plane $\Uppi$. In this case we were still able to provide lower bounds on the parameters of the obtained codes exploiting their geometric properties. Nevertheless, we could still show that one round of the bit-flipping algorithm has the best asymptotic performance in terms of error-correction capability for the given parameters of the defining parity-check matrix.

\section*{Acknowledgment}

The work of A. Neri was supported by the Swiss National Science Foundation through grant no.~187711.
The work of J. Rosenthal was supported by the Swiss National Science Foundation through grant no.~188430.







%
%
\bibliographystyle{abbrv}
\bibliography{biblio}
\end{document}